\documentclass[pra,8pt,twocolumn,superscriptaddress,floatfix,showpacs]{revtex4-2}
\usepackage{tikz}
\usepackage{mathtools}
\usetikzlibrary{arrows.meta, positioning,shapes,calc}
\usepackage[utf8]{inputenc}
\usepackage[T1]{fontenc}     %Output what you want e.g., é, l, a, ü
\usepackage[british]{babel}  %Do hyphenation according to british english
\usepackage[sc,osf]{mathpazo}
\usepackage[scaled=0.86]{berasans}  % URL font that go well wtih palatino
\usepackage[colorlinks=true, citecolor=blue, urlcolor=blue]{hyperref}  %Hyperlinks (pink, green, blue)
\usepackage{graphicx} % Package to insert exteral figures
\usepackage{subfig}
\usepackage[babel]{microtype}  %Improves text justification
\usepackage{amsmath,amssymb,amsthm,bm,amsfonts,mathrsfs,bbm} %Usefull math packages
\usepackage{comment}

\usepackage{xspace}  %Useful to add space in macros
\usepackage{pgfplots}
\usepackage{xcolor,colortbl}
\def\ba{\begin{equation}}
	\def\ea{\end{equation}}
\def\bea{\begin{eqnarray}}
	\def\eea{\end{eqnarray}}
\def\ben{\begin{equation*}}
	\def\een{\end{equation*}}
\def\bean{\begin{eqnarray*}}
	\def\eean{\end{eqnarray*}}
\def\bma{\begin{mathletters}}
	\def\ema{\end{mathletters}}
\def\bi{\begin{itemize}}
	\def\ei{\end{itemize}}

\newcommand{\be}{\begin{equation}}
	\newcommand{\ee}{\end{equation}}

\newcommand{\kommentar}[1]{}

\newcommand{\forget}[1]{}

\newtheorem{theorem}{Theorem}

\newtheorem{corollary}{Corollary}[theorem]

\begin{document}
	
	\title{%Genuine Network Nonlocality Offers Better Security In Quantum Key Distribution Protocol\\
		%	Full Network Nonlocality Offers Better Security In Quantum Key Distribution\\
		%Role Of Full Network Nonlocality In Quantum Key Distribution
		Security Analysis of Ekert-type QKD Protocol Coming Under Cloning Attacks}
	%#####################################
	\author{Kaushiki Mukherjee}
	\email{kaushiki.wbes@gmail.com}
	\affiliation{Department of Mathematics, Jhargram Raj College,Jhargram, West Bengal 721507, India.}
	\author{Sumana Karmakar}
	\email{sumanakarmakar88@gmail.com}
	\affiliation{Department of Mathematics, Chandrapur College, Chandrapur, Purba Bardhaman, West Bengal-713145, India.}
	%########################################
	
	\begin{abstract}
		We analyze the security of an Ekert-type bipartite key distribution protocol coming under cloning attacks by an un-trusted observer. Precisely, an eavesdropper is assumed to intercept and apply an asymmetric, phase-variant cloner on one or both qubits which are transmitted over public channel from an authenticated source to the trusted parties. We partially characterize two-qubit states from the perspective of designing a QKD protocol that can generate a distillable secret key among the trusted users even under such cloning attacks. We rely upon Devetak–Winter framework for computing key rate so as to identify states, which when used in the protocol, secret-key can be distilled. Utility of an entangled state, for designing a successful QKD protocol, changes from single-qubit to double-qubits cloning attacks. Interestingly, we can distinguish some useful families of two-qubit entangled states when only one of the transmitted qubits is cloned by a malicious third party.
		% For practical relevance, we further investigate detection efficiency of our protocol. 
	\end{abstract}
	%$4$-partite 
	%\pacs{03.65.Ud, 02.50.Le, 03.67.Ac}
	%\pacs{03.65.Ud, 03.67.Mn}
	\maketitle
	
	%########################################
	
	\section{Introduction}\label{intro}
	Cryptography is one of the branches of information science where quantum theory offers promising advantages over traditional classical theory\cite{QC1,wies,QC2,retsispg10}. Such advantages basically stem from non-classical behavior, more precisely, intrinsic randomness of quantum particles\cite{bb84,eck,cr,cr1,cr2,cr3,cr4,cr5,cr6}. Study of quantum non-classical correlations from perspective of cryptography thus plays a pragmatic role in developing theoretical base for quantum cryptography(QC). 
	\par Over last decade, with advancement of quantum technology, QC has undergone rapid development. Analogous to its classical counterpart, generation and distribution of private key in between trusted parties distribution is an inherent constitution of QC. Any quantum key distribution(QKD) protocol generates a private key shared between trusted parties who can communicate among themselves over public channels\cite{retsispg10}.
	\par \textit{BB84 protocol}, designed by Bennett and Brassard pioneered the study of QC\cite{bb84}. Since then study of QC has undergone multifaceted exploration. Existing quantum key distribution(QKD) protocols can be divided in two classes:\textit{preparation and measurement type}\cite{bb84,cr,cr1,cr2,cr3,cr4,cr5,cr6} and \textit{entanglement assisted type}\cite{eck,eck3,eck1,eck2,eck4,eck6,eck5,eck7,eck11,eck12}. While BB84 protocol belongs to the first class, entanglement assisted type protocol was first introduced by Ekert\cite{eck}. Since then many other QKD protocols have been designed following Ekert's protocol\cite{eck3,eck1,eck2,eck4,eck6,eck5,eck7,eck14,eck13}. Advantage of this class of protocols over classical ones is due to entanglement of quantum particles shared in between the trusted parties($A$ and $B$,say). In any such QKD protocol, $A$ and $B$ share multiple copies of an entangled state\cite{hor1} over which they perform some local measurements in mutually unbiased bases(MUBs\cite{mub1}). They use a fraction of the measurement outcomes for validation check.  Remaining outcomes are next used for public reconciliation of their measurement bases and then forming a raw key\cite{bran1} from the outcomes of identical MUBs. Finally the raw key is used to extract a secure key via suitable information reconciliation\cite{eck8i} and privacy amplification\cite{eck8i} strategies. 
	\par Ideally pure entangled state is supposed to be distributed from a source to $A,B$(trusted parties). However, in practical situations, an un-trusted party(Eve,say) can have control over the source, communication channels and even the measurement devices of $A$ and $B.$ Adversarial control over the QKD protocol hampers the security of the private key. It was first pointed out be Ekert\cite{eck} that $A$ and $B$ can detect any malicious third party via violation of Bell-CHSH inequality. Precisely, absence of a Bell inequality's violation indicates  eavesdropper's presence. Since its inception, such Bell inequality's violation based verification strategy has been exploited to design several secure QKD protocols\cite{eck1,eck2,eck3,eck4,eck5,eck6,eck7,main,selfq,selfq1}. 
	\par Over last few years, some research works have provided state-dependent analysis of QKD protocols where the authors have also exploited Bell-type inequality violation for minimizing Quantum Bit Error Rate(QBER) generated in such protocols\cite{main,selfq,selfq1}. Precisely, these works have characterized two-qubit state space based on their utility in designing entanglement assisted type QKD protocols considering Bell type inequality violation and QBER lying below some threshold limit as two check steps. However, these two check steps, being only necessary to guarantee absence of any adversary, can only act as validation checks and therefore cannot provide unconditional security. Consequently, it may happen that no secure private key can be ultimately distilled using the protocol. In this context it becomes pertinent to refine existing state space characterization so as to identify states useful for designing QKD protocol from which secure key can be distilled. But for that one needs to consider an adversary-aware analysis of such QKD protocols. Present work will be contributory in this direction.
	\par Let us consider a bipartite QKD protocol where Eve has no control over the source($\mathbf{S}$), distributing entangled particles, but adopts some cloning strategy(to be discussed later) over one of the qubits of the entangled state distributed by $\mathbf{S}$ in between two trusted users $A,B.$ For a particular instance, let a two-qubit state $\rho$ with correlation tensor $\small{T}$$=$$\textmd{\small{diag}}(0.8,$$-$$0.8,0.82)$ and with both local Bloch vectors given by $(0,0,0.18)$ be shared between $A,B.$ As per characterization provided in \cite{main}, Bell-CHSH violation is observed in QKD protocol using $\rho$ and also QBER is less than $14\%.$ So the protocol using $\rho$ is not aborted. However, no bipartite secret key can be distilled from the protocol as Devetak-Winter key rate\cite{dw1} turns out to be negative($-$$0.11024$bits). So $\rho$ cannot be used to design QKD protocol generating a bipartite secure key though the protocol passes both the validation checks. Above illustration clearly indicates the necessity of further analyzing two-qubit state space so as to segregate states based on their utility to design protocols via which private key can be distilled. This motivates the present work.
	\par Comparison of trusted parties information content with that of malicious party is the most obvious way to verify security in a QKD protocol\cite{eck18}. With such a concept in hindsight, we consider $1$-way secret key distillation criterion for providing an adversary-aware analysis of arbitrary two-qubit states in context of designing QKD protocol\cite{main}. In general, for any un-trusted party having control over the source, the worst possible situation occurs when Eve holds purification of the state shared between the trusted parties. However, characterizing states in such a general scenario becomes complicated for obvious reasons. Here we consider an explicit quantum adversarial framework where Eve can clone one or both qubits distributed to the trusted users. Intending to add an operational interpretation of existing geometric characterization of states\cite{main} by using some practically relevant parameters such as $1$-way secret-key rate, quantity of information leakage, here we refrain ourselves to a specific type of eavesdropping strategy. 
	\par  \textit{QKD Protocol:} We explicitly use the $1$-way key distillation protocol($\mathcal{P}_{QKD},$say) framed in \cite{main}. Let us first briefly state the steps of $\mathcal{P}_{QKD}$:
	\begin{enumerate}
		\item \textit{Qubits Distribution Stage:} A source($\mathbf{S}$) distributes $n$ copies of a bipartite entangled state $\varrho_{AB}$ between $A,B.$
		\item \textit{Measurement Stage:} Each of $A,B$ performs the following measurements:
		\begin{enumerate}
			\item Each of $A,B$ measure $h$$<$$n$ copies of single qubit of $\varrho_{AB}$ in one of two randomly chosen two-dimensional MUBs. 
			\item On each of remaining $n$$-h$ copies each of $A,B$ performs projective measurements randomly in any one of two arbitrary directions.
		\end{enumerate}
		\item \textit{Measurement Setting Disclosure Stage:}
		They broadcast the type of measurement(step 2(a) or (2b)) performed in each of $n$ rounds. $A,B$ only keep the rounds in which the measurement types match. Let $s_1,s_2$ denote the successful rounds corresponding to first type and second type of measurement. Then 
		\begin{eqnarray*}
			s_1&<&h\,\,\textmd{and } s_2<n-h.
			%r_2&<&\textmd{Min}_{i=1}^4n-h_i=n-\textmd{Max}_{i=1}^4h_i.
		\end{eqnarray*}
		\item \textit{Bell-CHSH Inequality Testing Stage:} Using bipartite measurement statistics resulting from local projective measurements in $s_2$ rounds, $A,B$ test Bell-CHSH inequality. In case violation is observed next step is executed.Otherwise it is aborted.
		\item \textit{Sifting Stage:} $A,B$ publicly announce their chosen MUBs for each of $s_1$ rounds in step (2a). Particularly:
		\begin{itemize}
			\item \textit{Bases Reconciliation:} $A,B$ discard rounds in which their MUBs do not match. Let $r_1$$<$$s_1$ denote the rounds where $A,B$ have chosen identical MUBs.
			\item \textit{Sifting Keys Generation:} Each MUB being two-dimensional, the outputs of the identical MUBs of each of $A,B$ form a bit string of length $r_1.$ These bit strings are often referred to as the \textit{raw/sifted keys.}
			\item \textit{QBER Generation:} $\varrho_{AB}$ is an entangled state. Hence, ideally, in each of $r_1$ rounds in the sifting stage, $A,B$ are supposed to obtain correlated outputs from local measurements in their identical MUBs. But this does not happen in practical situations leading to generation of QBER($\mathcal{Q}$) in the protocol. $A,B$ use a fraction of their sifted keys to compute $\mathcal{Q}$. If $\mathcal{Q}$$>$$\mathcal{Q}_0$$=$$0.1414$\cite{main} then the protocol is aborted. Otherwise next step is executed.
		\end{itemize}
		\item \textit{Generation of Secret Key:} W.L.O.G., let $A$ process her raw key and send the required classical information to $B$ for error reconciliation. Based on the information received $B$ correct discrepancies between their raw keys. Finally they apply some suitable privacy amplification procedures to the reconciled keys thereby extracting a shorter but secure secret key by suppressing sufficient amount of any un-trusted party's side information.
	\end{enumerate}
	\textit{Validation Checks in $\mathcal{P}_{QKD}$:} For any state $\rho_{AB}$ shared between $A,B$(trusted parties), $\mathcal{P}_{QKD}$ passes the first validation check in $4^{th}$ step if:
	\begin{eqnarray}\label{vald1}
		\lambda_1^2+\lambda_2^2&>&1.
	\end{eqnarray}
	$\lambda_1,\lambda_2$ denote the largest $2$ singular values of $\rho_{AB}^{'}$s correlation tensor. $\mathcal{P}_{QKD}$ passes the second validation check in $5^{th}$ step if\cite{main}:
	\begin{eqnarray}\label{vald2}
		\lambda_1+\lambda_2&>&\sqrt{2}.
	\end{eqnarray}
	\textit{Key Distillation In Presence of Adversary:} Now, let $\mathcal{P}_{QKD}$ comes under some cloning attacks of a malicious third party($E$). Let $\varrho_{ABE}$ be the resulting tripartite state in $\mathcal{P}_{QKD}$. Clearly, $\varrho_E$$=$$\textmd{Tr}_{AB}[\varrho_{ABE}]$ is the state with $E.$ In presence of $E,$ the \textit{Devetak-Winter key rate}($r_D$,say) determines the asymptotic amount
	of secret key that $A,B$ can extract from their raw keys\cite{dw1,dw2} in the last step of $\mathcal{P}_{QKD}$. Above protocol being $1$-way reconciliation type(from $A$ to $B$), Devetak-Winter key rate is given by\cite{dw1}:
	\begin{eqnarray}\label{wr7}
		r_D&=&\frac{1}{3}\sum_{j=1}^3(I_j(A:B)-\chi_j(A:E)).
	\end{eqnarray}
	In Eq.(\ref{wr7}) $I_j(A$$:$$B)$ denote the mutual information between $A,B$ when they measure in $j^{th}$ MUB on $\rho_{AB}^{'}$$=$$\textmd{Tr}_{E}[\varrho_{ABE}]$. MUBs being $2$-dimensional, up to local unitaries, there are only 3 distinct MUBs\cite{mub1}. So $j$ runs over $1,2,3.$ \\
	$\chi_j(A:E)$ denote the Holevo quantity\cite{dw2} limiting $E^{'}$s side information(for any POVM $M_{\mathbf{e}}$ that it measures on $\textmd{Tr}_{AB}[\varrho_{ABE}]$) when the trusted parties have measured in $j^{th}$ MUB:
	\begin{eqnarray}\label{wr3}
		I_j(A:E_{\mathbf{e}})&\geq& \chi_j(A:E)\nonumber\\
		\textmd{\small{where }} \chi_j(A:E)&=&S(\varrho_E)-\sum_{a_j}p(a_j)S(\varrho_E^{(a_j)}
	\end{eqnarray} 
	In Eq.(\ref{wr3}),$ E_e$ stands for the classical variable corresponding to any possible $\{M_e\}$ and $S$ denotes von-neumann entropy\cite{nie}. $\varrho_E^{(a_j)}$ is $E^{'}$s conditional state when $A$ obtains output $a_j$ with probability $p(a_j)$ while measuring her subsystem in $j^{th}$ MUB. \\
	Now as $r_D$ gives the asymptotic amount
	of secret key, $A,B$ can distill a secret key by $\mathcal{P}_{QKD}$ only if\cite{dw1,dw2}: 
	\begin{eqnarray}\label{keyd}
		r_D&>& 0.
	\end{eqnarray}
	Violation of Eq.(\ref{keyd}) implies absence of positive asymptotic key rate\cite{dw1,dw2} and consequently in such case no secret key can be ultimately shared between the trusted parties. This in turn will render $\varrho_{AB}$ to be useless in designing a secure QKD protocol($\mathcal{P}_{QKD}$) although $\mathcal{P}_{QKD}$ will pass both the validation checks(steps $4,5$).  \\
	\textit{Single-Qubit Cloning Strategy By Eve($E$):} $E$ clones one of the qubits distributed from $\mathbf{S}$ to $A$ and $B.$ Precisely, without loss of any generality, let $E$ intercept the qubit sent to $B$ and apply the following asymmetric phase variant cloning\cite{clone}:
	\begin{eqnarray}\label{cln1}
		U_C(|i0\rangle)&=&\cos\frac{\theta}{2}	|i0\rangle+i\sin\frac{\theta}{2} |i\oplus 1 0\rangle,\,\theta\in[0,\frac{\pi}{4}].
	\end{eqnarray}
	Let $\rho_{ABE}$ denote the $3$-qubit state resulting after above operation(Eq.(\ref{cln1})) is applied on $2^{nd}$ qubit of arbitrary two-qubit  $\rho_{AB}$$=$$\sum_{\substack{j,k,\\l,m=0}}^1a_{j,k,l,m}|jk\rangle\langle lm|.$ $E$ keeps the cloned qubit, i.e., $3^{rd}$ qubit of $\rho_{ABE}$(see supplementary material for details) and sends the original qubit($2^{nd}$ of $\rho_{ABE}$) to $B.$ $E$ maintains correlations with the trusted parties through the state $\rho_{ABE}.$ Trusted parties $A,B$ execute $\mathcal{P}_{QKD}$ using  $\rho_{AB}^{'}$$=$$\textmd{Tr}_E[\rho_{ABE}].$ Let $\mathcal{P}_{QKD}^{C_1}$ denote the QKD $\mathcal{P}_{QKD}$ when it comes under the single-qubit cloning attack specified by Eq.(\ref{cln1}).
	\par In $\mathcal{P}_{QKD}^{C_1},$ $E$ retains the state $\rho_E$$=$$\textmd{Tr}_{AB}[\rho_{ABE}].$ Using $\rho_E$ and conditional states of $E$ resulting from $A^{'}$s projective measurement(in anyone of three possible MUBs), $E$ gains side information from $\mathcal{P}_{QKD}^{C_1}.$ Depending on $E^{'}$s information gain and the mutual information shared between $A,B$, secret key can be distilled from $\mathcal{P}_{QKD}^{C_1}$(involving $\rho_{AB}^{'}$) if Eq.(\ref{keyd}) is satisfied. This in turn renders the original state $\rho_{AB},$ distributed from $\mathbf{S},$ to be useful.
	\par Let $\mathbf{S}$ distribute pure state:
	\begin{eqnarray}\label{pure}
		\rho_{AB}&=&|\psi\rangle\langle \psi|\textmd{\small{with}}\\
		|\psi\rangle&=&\cos\alpha |00\rangle+\sin\alpha |11\rangle\,\,0\leq\alpha\leq0.785.
	\end{eqnarray}
	Interestingly, $\mathcal{P}_{QKD}^{C_1}$ can always be used to generate secret key. Precisely, we get the following result:
	\begin{theorem}\label{clo1}
		When $\mathbf{S}$ distribute pure two-qubit entangled state, secret key can always be distilled from $\mathcal{P}_{QKD}^{C_1}$.
	\end{theorem}
	\textit{Proof:} See supplementary material.\\
	For a numerical illustration of Theorem.\ref{clo1}, let us consider maximally entangled state and let $E^{'}$s cloning operation be specified by $\theta$$=$$\frac{\pi}{4}.$ Correlation tensor($T$,say) of $\rho_{AB}^{'}$ is $\textmd{\small{diag}}(1,$$-$$1,1).$ So, $\rho_{AB}^{'}$ satisfy both Eqs.(\ref{vald1},\ref{vald2}). Now $I(A:B)$$=$$0.739127$ and Holevo quantity $\chi_(A:E)$$=$$0.1405.$ So $E^{'}$s side information cannot be more than $0.1405$ bits. Devetak-Winter key rate is thus $r_D$$=$$0.5986$ bits(see supplementary material). So by Eq.(\ref{keyd}), a secret key can be distilled from $\mathcal{P}_{QKD}^{C_1}.$ \\
	\textit{Werner Class of States Always Useful In $\mathcal{P}_{QKD}^{C_1}$ :}In practical situations, noisy entangled state gets distributed from $\mathbf{S}$ to the trusted parties. For example, let $\rho_{AB}$ be a Werner state\cite{hor1}:
	\begin{eqnarray}\label{werner}
		\rho_{AB}&=&v|\psi\rangle\langle \psi|+\frac{1-v}{4}\mathbb{I}_4,\textmd{ \small{where}}\\
		|\psi\rangle&=&\frac{|01\rangle-|10\rangle}{\sqrt{2}}\nonumber
	\end{eqnarray}
	$2^{nd}$ qubit of $\rho_{AB}$ comes under $E^{'}$s cloning attack. Even in this case secret key distillation becomes possible. Precisely, we can say:\textit{
		When both qubits of a maximally entangled state get distributed from $\mathbf{S}$ across independent and identical single-qubit depolarizing channel($\mathcal{N}(\sqrt{v})$) in $\mathcal{P}_{QKD}^{C_1},$ secret key can always be distilled.}\\
		% provided $\mathcal{P}_{QKD}^{C_1}$ passes both validation checks.
Proof of the result provided in supplementary material.\\
\par For a particular numerical illustration of above result, let us consider $\mathcal{P}_{QKD}^{C_1}$ where each qubit of maximally entangled state is communicated over $\mathcal{N}(\sqrt{0.8})$ and let $E^{'}$s cloning operation be specified by $\theta$$=$$\frac{\pi}{6}.$ Here $T$ of $\rho_{AB}^{'}$ is $\textmd{\small{diag}}($$-$$0.8,$$-$$0.8).$ $\rho_{AB}^{'}$ thus satisfies both Eqs.(\ref{vald1},\ref{vald2}). Now $I(A:B)$$=$$0.478009$ and $\chi_(A:E)$$=$$0.0449767.$ So $E^{'}$s gains at most $0.0449767$ bits of information from $\mathcal{P}_{QKD}^{C_1}$ and Devetak-Winter key rate turns out to be $0.433033$ bits. So by Eq.(\ref{keyd}), a secret key can be distilled from $\mathcal{P}_{QKD}^{C_1}.$
\par Above result indicates possibility of distilling secret key from $\mathcal{P}_{QKD}^{C_1}$ when qubits of pure entangled state are communicated over noisy channels. However, apart from noisy communication, there may be several other factors resulting in distribution of noisy entangled states from $\mathbf{S}.$ So, in general it becomes relevant to consider $\rho_{AB}$ to be an arbitrary two-qubit state\cite{gam}:
	\begin{eqnarray}\label{st41}
		\small{\rho}_{AB}=\small{\frac{1}{4}(\mathbb{I}_{2}\times\mathbb{I}_2+\vec{\mathbf{m}}.\vec{\sigma}\otimes \mathbb{I}_2+\mathbb{I}_2\otimes \vec{\mathbf{n}}.\vec{\sigma}+\sum_{i=1}^{3}t_{i}\sigma_{i}^{\otimes 2}}.
	\end{eqnarray}
	Impossibility of deriving any closed form for $r_D$ makes it impossible to provide a complete characterization of arbitrary two-qubit state $\rho_{AB}$(Eq.(\ref{st41})) in context of its utility to design a successful $\mathcal{P}_{QKD}^{C_1}$. However, if we confine to a sub-class of $\rho_{AB}$ having no local Bloch vectors, i.e., $\vec{\mathbf{m}}$$=$$\vec{\mathbf{n}}$$=$$\vec{\Theta},$ we get an interesting result.
	\begin{theorem}\label{clo3}
		When $\mathbf{S}$ distributes arbitrary two-qubit entangled state having no local Bloch vector, secret key can always be distilled from $\mathcal{P}_{QKD}^{C_1}.$
		% provided it passes both validation checks.
	\end{theorem}
	\textit{Proof:} See supplementary material.\\
	Theorem.\ref{clo3} thus provides partial characterization of arbitrary $\rho_{AB}$ in terms of its correlation tensor. The class of two-qubit states with no local Bloch vectors is more commonly known as \textit{states with mixed marginals}\cite{hor1}. This class encompasses several important families of noisy entangled states resulting from either erroneous generation of entangled states or distribution of pure entangled states over noisy channels. We get the following results as a direct application of Theorem.\ref{clo3}.
	\begin{corollary}
	When mixed entangled state, generated from application of faulty Hamiltonian and $\mathcal{C-NOT}$ gates, gets transmitted from $\mathbf{S},$ secret key can always be distilled from $\mathcal{P}_{QKD}^{C_1}.$ 	
	\end{corollary}
	Such a two-qubit entangled state has the form:
	\begin{eqnarray}
		\rho_{AB}&=&\frac{1+p_1p_2}{4}\sum_{i=0}^1|ii\rangle\langle ii|+\frac{1-p_1p_2}{4}\sum_{i,j=0}^1|ij\rangle\langle ij|\nonumber
		\\
		&-&\frac{p_1p_2}{2}\sum_{i,j=0}^1|ii\rangle\langle jj|.
	\end{eqnarray}
	Here $p_1$ and $p_2$ denote the noise parameters characterizing the faulty Hamiltonian gate and $\mathcal{C-NOT}$ gate(with $1^{st}$ qubit as control qubit) used to generate entanglement from $|01\rangle\langle 01|$\cite{per1,per}. One can easily check that this class of states has no local Bloch vector. 
	\begin{corollary}
When both qubits of a maximally entangled state get distributed from $\mathbf{S}$ across independent and identical phase damping channel($\mathcal{N}(q_1)$) or bit-flip channel($\mathcal{N}(q_2)$) in $\mathcal{P}_{QKD}^{C_1},$ secret key can always be distilled.
\end{corollary}
Above observations point out that even if entanglement gets degraded either due to faults in entangling gates or due to transmission along noisy channels, resulting noisy entangled states can still be used to design successful $\mathcal{P}_{QKD}^{C_1}.$
%(details expressions of the states provided in supplementary material).
	%This class of states encompasses the noisy entangled generated from $\mathbf{S},$ resulting from imperfection in Hadamard and $C$$-$$NOT$ gates applied over $|01\rangle\langle 01|.$(ref persstncy \cite{indr5}):
	\par If we consider $\rho_{AB}$ to have local Bloch vector, $\mathcal{P}_{QKD}^{C_1}$ may not always be used for key distillation. For example let us consider that $\mathbf{S}$ distribute a sub-class of the $X$ states:
	\begin{eqnarray}\label{x}
		\rho_{AB}&=&	\left[ {\begin{array}{cccc}
				x_1&0&0&0\\
				0&x_2&y&0\\
				0&y&x_{3}&0\\
				0&0&0&0\\		
		\end{array} } \right]\\
		\textmd{\small{with}}&&x_1+x_2+x_3=1;\,y^2\leq x_2\cdot x_3.\nonumber
	\end{eqnarray}
	For this state(Eq.(\ref{x})), the Bloch parameters(Eq.(\ref{st41})) are:
	\begin{eqnarray}
		T&=&\textmd{diag}(2y,2y,x_1-x_2+x_3),\nonumber\\
		\vec{\mathbf{m}}&=&(x_1+x_2-x_3,0,0);\,\,\vec{\mathbf{n}}=(0,0,x_1-x_2-x_3)
	\end{eqnarray}    
	Let $2^{nd}$ qubit of $\rho_{AB}$(Eq.(\ref{x})) be subjected to $E^{'}$s cloning attack(Eq.(\ref{cln1})) specified by $\alpha$$=$$\frac{\pi}{4}.$ In this case secret key cannot always be distilled from $\mathcal{P}_{QKD}^{C_1}$(for some members from this subclass) even if the protocol passes both the validation checks(see Fig.\ref{fig1}). 
	%For a particular numerical consider 
	\begin{center}
		\begin{figure}
			\includegraphics[width=2.9in]{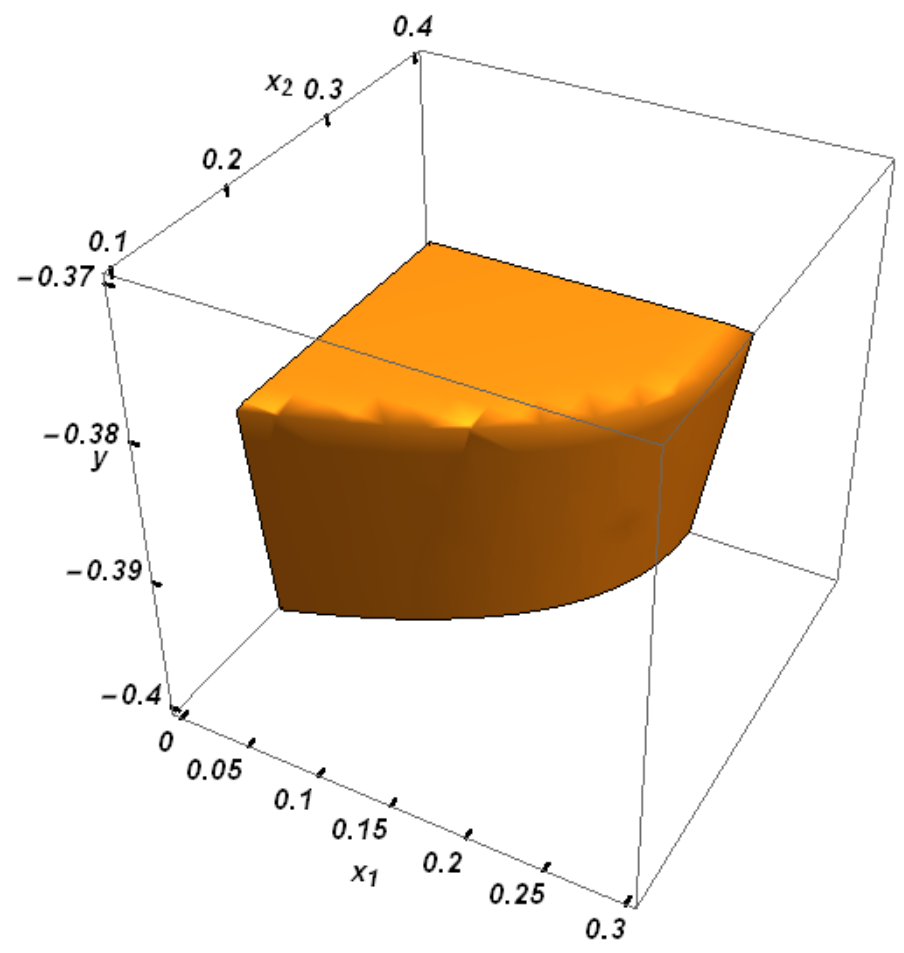} 
			\caption{\emph{Shaded region gives a subspace in the parameter space $(x_1,x_2,y)$ of $\rho_{AB}$ specified by Eq.(\ref{x}). When $\rho_{AB}$ corresponding to any point in the region is used in $\mathcal{P}_{QKD}^{C_1}$ with $\theta$$=$$\frac{\pi}{4}$, the protocol passes both the validation checks but $r_{D}$ turns out to be negative. So $\mathcal{P}_{QKD}^{C_1}$ cannot be used to distill any secret key though the trusted parties do not abort the protocol.}}
			\label{fig1}
		\end{figure}
	\end{center}
For further illustration, let us consider the following family of two-qubit states:
	\begin{eqnarray}\label{st411}
	\small{\rho}_{AB}&=&\frac{1}{4}(\mathbb{I}_{2}\times\mathbb{I}_2+(0,0,\sqrt{\small{s_2^2-s_1^2}}).\vec{\sigma}\otimes \mathbb{I}_2\nonumber\\
	&+&\mathbb{I}_2\otimes (0,0,\sqrt{\small{s_2^2-s_1^2}}).\vec{\sigma}+s_{1}\sigma_{1}^{\otimes 2}\nonumber\\
		&-&s_{1}\sigma_{2}^{\otimes 2}+s_{2}\sigma_{3}^{\otimes 2})\,\, 0\leq s_1,s_2\leq 1.
\end{eqnarray}
Several members from this family(Eq.(\ref{st411})) cannot be used to design successful $\mathcal{P}_{QKD}^{C_1}$(see Fig.\ref{extra}).
	\begin{center}
	\begin{figure}
		\includegraphics[width=2.8in]{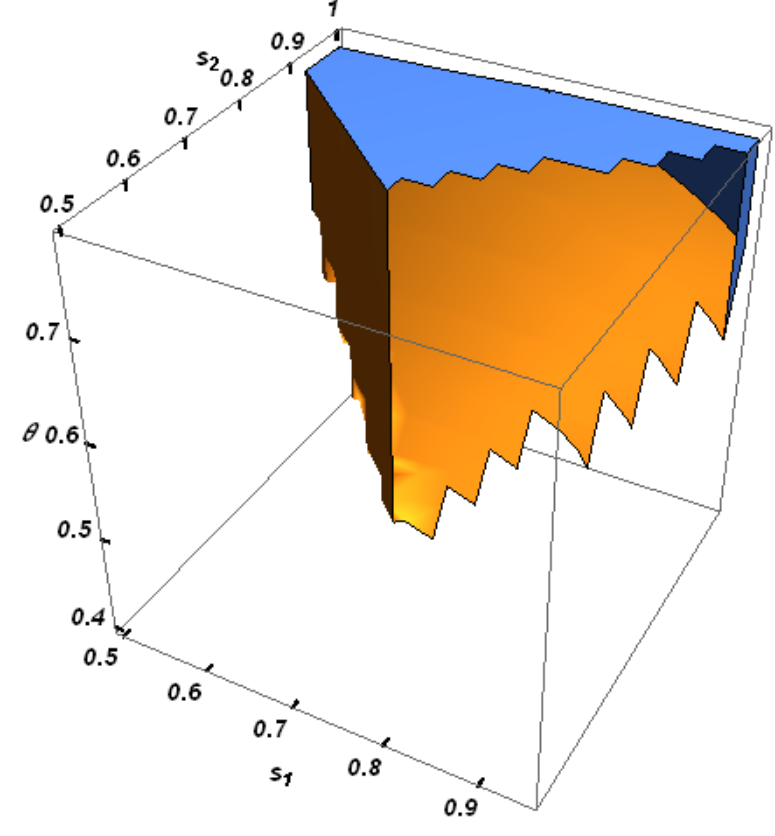} 
		\caption{\emph{Shaded region indicates a subspace in the  parameter space $(s_1,s_2,\theta)$ for which two-qubit states(corresponding to $(s_1,s_2)$ in this region) from the family specified by Eq.(\ref{st411}) cannot be used to design successful $\mathcal{P}_{QKD}^{C_1}$ specified by cloning parameter $\theta$ corresponding to point $(s_1,s_2,\theta)$ lying in this region although the protocol passes both the validation checks.}}
		\label{extra}
	\end{figure}
\end{center}
Above demonstration may give rise to the intuition that perhaps entangled state with local Bloch vectors may be useless in designing successful $\mathcal{P}_{QKD}^{C_1}.$ Our next result turns out to be counter-intuitive in this regard.
	\begin{theorem}
When both qubits of a maximally entangled state get distributed from $\mathbf{S}$ across independent and identical amplitude damping channel($\mathcal{N}(q_3)$) in $\mathcal{P}_{QKD}^{C_1},$ secret key can always be distilled.		
	\end{theorem}
\textit{Proof:} This can be proved following similar approach as that used for proving Theorem.\ref{clo1}.\\
%Form of the state along with its Bloch parameter details are discussed in supplementary material.
\textit{Both Qubits Cloning Attack:} Let us now consider the situation where $E$ intercepts both the qubits of $\rho_{AB}$(transmitted from $\mathbf{S}$) and applies the cloner specified by Eq.(\ref{cln1}) on both the individual qubits. Let $\rho_{AE_1BE_2}$ denote the resulting state.
	$E$ keeps the cloned qubits, i.e., $2^{nd},$ $4^{th}$ qubits of $\rho_{AE_1BE_2}$ and sends the original qubits, i.e., $1^{st},$ $3^{rd}$ qubits of $\rho_{AE_1BE_2}$ to $ A,B$ respectively. The trusted parties execute $\mathcal{P}_{QKD}$ using $\rho_{AB}^{''}$$=$$\textmd{Tr}_{E_1E_2}[\rho_{AE_1BE_2}]$ whereas $E$ retains a two-qubit state $\rho_{E_1E_2}$$=$$\textmd{Tr}_{AB}[\rho_{AE_1BE_2}]$. Let $\mathcal{P}_{QKD},$ coming under $E^{'}$s double-qubit cloning attack be denoted by $\mathcal{P}_{QKD}^{C_2}.$ 
	\par Compared to $\mathcal{P}_{QKD}^{C_1},$ in $\mathcal{P}_{QKD}^{C_2}$ $E$ can now establish stronger correlations between $\rho_{E_1E_2}$ and the trusted parties' measurement outcomes obtained from $\rho_{AB}^{''}.$ Consequently, $E$ will gain more side information in $\mathcal{P}_{QKD}^{C_2}.$ This in turn reduces the chance of distilling secret key. 
	\par In this context, let us recall Theorem.\ref{clo1}. Any pure state can be used to design a successful $\mathcal{P}_{QKD}^{C_1}.$ But we observe violation of this theorem in $\mathcal{P}_{QKD}^{C_2}$. Precisely, there exist pure two-qubit entangled states, which when used in $\mathcal{P}_{QKD}^{C_2},$ secret key cannot be distilled(see Table.\ref{table:ta1}). Similar violation of theorem.\ref{clo3} is also observed(see Table.\ref{table:ta1}). However, $E$ does not gain complete knowledge about the key via its double cloning strategy. Precisely, there exist entangled states which can be used to design a successful $\mathcal{P}_{QKD}^{C_2}$(see Table.\ref{table:ta1}). 
	%stablishing correlations between her retained system and Alice’s and Bob’s outcomes,

	%Using $\rho_E,$ $E$ gains side information from $\mathcal{P}_{QKD}.$ Depending on $E^{'}$s information gain and the mutual information shared between $A,B$, secret key can be distilled from $\mathcal{P}_{QKD}$(involving $\rho_{AB}^{'}$) if Eq.(\ref{keyd}) is satisfied. This in turn renders the original state $\rho_{AB},$ distributed from $\mathbf{S},$ to be useful.
	\begin{table}
		\caption{Few numerical instances provided for comparing $\mathcal{P}_{QKD}^{C_1} $ and $\mathcal{P}_{QKD}^{C_2}.$ Clearly, $E$ obtains more side information from $\mathcal{P}_{QKD}^{C_2}.$ Moreover, some $\rho_{AB}$ can  be used to design a successful $\mathcal{P}_{QKD}^{C_1}$ only.}
		\[
		\begin{array}{|c|c|c|c|c|}
			\hline
			\rho_{AB}&\small{Protocol}&\chi&r_{D}&\textmd{\small{Remarks}}\\ \hline
			Eq.(\ref{pure})  & \mathcal{P}_{QKD}^{C_1}  &0.2258 & 0.3356&\textmd{\small{Both successful; }} \\ 
			\textmd{ \small{with}}& \textmd{\small{with} }\theta=\frac{\pi}{3}&\small{bits}&\small{bits}&\textmd{ \small{More informa-}}\\\cline{2-4}
			\alpha=\frac{\pi}{5}	&\mathcal{P}_{QKD}^{C_2}\textmd{ \small{with}} &0.3103 & 0.0604&\textmd{\small{tion leakage}}  \\
			& \theta_1=\theta_2=\frac{\pi}{3}&\small{bits}&\small{bits}&\textmd{\small{from} }\mathcal{P}_{QKD}^{C_2}\\\cline{2-4}
			\hline
			Eq.(\ref{pure})  & \mathcal{P}_{QKD}^{C_1}  &0.0903 & 0.4277&\textmd{\small{Only }}\mathcal{P}_{QKD}^{C_1} \\ 
			\textmd{ \small{with}}& \textmd{\small{with} }\theta=\frac{\pi}{4}&\small{bits}&\small{bits}&\textmd{ \small{successful}}\\\cline{2-4}
			\alpha=\frac{\pi}{3}	&\mathcal{P}_{QKD}^{C_2}\textmd{ \small{with}}  &0.4175 & -0.0365& \\
			&\theta_1=\theta_2=\frac{\pi}{4}&\small{bits}&\small{bits}&\\\cline{2-4}
			\hline
			Eq.(\ref{werner})  & \mathcal{P}_{QKD}^{C_1}  &0.0392 &0.3729 &\textmd{\small{Both successful; }} \\ 
			\textmd{ \small{with}}& \textmd{\small{with} }\theta=\frac{\pi}{3}&\small{bits}&\small{bits}&\textmd{ \small{More informa-}}\\\cline{2-4}
			v=0.75	&\mathcal{P}_{QKD}^{C_2}\textmd{ \small{with}} &0.1708 & 0.2029&\textmd{\small{tion leakage}}  \\
			&\theta_1=\theta_2=\frac{\pi}{6}&\small{bits}&\small{bits}&\textmd{\small{from} }\mathcal{P}_{QKD}^{C_2}\\\cline{2-4}
			\hline
			Eq.(\ref{werner})  & \mathcal{P}_{QKD}^{C_1}  &0.2574 & 0.0655&\textmd{\small{Only }}\mathcal{P}_{QKD}^{C_1} \\ 
			\textmd{ \small{with}}& \textmd{\small{with} }\theta=\frac{\pi}{4}&\small{bits}&\small{bits}&\textmd{ \small{successful}}\\\cline{2-4}
			v=0.71	&\mathcal{P}_{QKD}^{C_2}\textmd{ \small{with}}  &0.2777 & -0.0151& \\
			& \theta_1=\theta_2=\frac{\pi}{4}&\small{bits}&\small{bits}&\\\cline{2-4}
			\hline
			
		\end{array}
		\]
	\end{table}\label{table:ta1}
\textit{Concluding Remarks:} Present work can be considered as an eavesdropper aware analysis of a bipartite entanglement-assisted QKD protocol for cloning strategy adopted by the third party. We use Devetak-winter key rate as a tool for security analysis of such protocol. Such analysis provides with partial characterization of two-qubit states in context of designing QKD protocol from which we can distill secret key. By doing so we have been able to identify some useful family of states in this regard when only one of the transmitted qubits gets cloned. Completely characterizing the entire two-qubit space can be considered as a potential direction of future research. When both the transmitted qubits get cloned, many of the states that were useful for designing successful QKD under single-qubit cloning attacks, turn out to be useless. Providing a concrete formalism of distinguishing states useful under these two types of external attacks deserves future investigation. Also, one may consider some more general eavesdropping strategies so as to identify states which can be used to design successful Ekert-type protocols coming under those attacks. 

	\section*{Supplementary Material}
	\subsection{Mutual Information $I(A:B)$ In Terms Of Bloch Parameters}
	Let $A,$ $B$ share an arbitrary two-qubit state $\rho$(Eq.(\ref{st41})):
	\begin{eqnarray*}
		\rho&=&\small{\frac{1}{4}(\mathbb{I}_{2}\times\mathbb{I}_2+\vec{\mathbf{m}}.\vec{\sigma}\otimes \mathbb{I}_2+\mathbb{I}_2\otimes \vec{\mathbf{n}}.\vec{\sigma}+\sum_{i=1}^{3}t_{i}\sigma_{i}\otimes\sigma_{i})}.
	\end{eqnarray*}
	Here $\vec{\mathbf{m}}$$=$$(m_1,m_2,m_3)$ and $\vec{\mathbf{n}}$$=$$(n_1,n_2,n_3).$\\
	Here we derive the expression of $I(A:B)$ in terms of parameters of $\rho.$\\
	Now, as discussed in main text, each of the trusted parties can measure in any of the three MUBs: $X,Y$ or $Z$ basis. So:
	\begin{eqnarray}\label{info1}
		I(A:B)&=&\frac{1}{3} \sum_{j=1}^3[I_j(A:B)]
	\end{eqnarray}
	where $I_1,I_2,I_3$ corresponds to mutual information between $A,B$ when they both measure in $X,Y$ and $Z$ basis respectively. \\
	\textit{Expression of $I_1(A:B):$} Let us consider measurement in $X$ basis. As we have fixed the basis for time being, to simplify the notations for deriving the expression in $X$ basis we drop the suffix in writing the entropy and probability terms. So,
	\begin{eqnarray}\label{info2}
		I_1(A:B)=H(A)-H(A|B) &&\nonumber\\
		=-\sum_{i=0}^1 p(a=i)\log_2(p(a=i))+\sum_{j=0}(p(b=j)H(A|B=j))&&\nonumber\\
		=-\sum_{i=0}^1 p(a=i)\log_2(p(a=i))+\sum_{j=0}(p(b=j)\times &&\nonumber\\
		\sum_{i=0}^1p(a=i|b=j)\log_2(p(a=i|b=j))&&\nonumber\\
		=-\sum_{i=0}^1 p(a=i)\log_2(p(a=i))+&&\nonumber\\
		\sum_{\substack{i,\\j=0}}^1p(a=i)p(b=j|a=i)\log_2(L)\,\textmd{\small{where}}&&\nonumber\\
		L=	\frac{p(b=j|a=i)p((a=i))}{\sum_{k=0}^1p(b=j|a=k)p(a=k)}&&\nonumber\\
	\end{eqnarray}
	$H(.)$ denotes Shanon entropy\cite{nie}.\\
	Now, the probability terms are:
	\begin{eqnarray}\label{info3}
		p(a=j)&=&\frac{1+(-1)^j m_1}{2};\nonumber\\
		p(b=j|a=i)&=&\frac{1}{2}+\frac{(-1)^{j}n_1+(-1)^{i+j} t_1}{2(1+(-1)^i m_1)}
	\end{eqnarray}
	Using Eq.(\ref{info3}), on simplification Eq.(\ref{info2}) gives:
	\begin{eqnarray}\label{info4}
		I_1(A:B)=H(\frac{1+m_1}{2})-\nonumber\\
		\sum_{j=0}^1\frac{1+(-1)^jn_1}{2} H(\frac{1-t_1+(-1)^j(n_1-m_1)}{2(1+(-1)^jn_1)})
	\end{eqnarray}
	Similarly,for $Y$ basis measurement and $Z$ basis measurements, we get respectively:
	\begin{eqnarray}\label{info5}
		I_2(A:B)=H(\frac{1+m_2}{2})-\nonumber\\
		\sum_{j=0}^1\frac{1+(-1)^jn_2}{2} H(\frac{1-t_2+(-1)^j(n_2-m_2)}{2(1+(-1)^jn_2)})\nonumber\\
		\textmd{and}\qquad\qquad\qquad\nonumber\\
		I_3(A:B)=H(\frac{1+m_3}{2})-\nonumber\\
		\sum_{j=0}^1\frac{1+(-1)^jn_3}{2} H(\frac{1-t_3+(-1)^j(n_3-m_3)}{2(1+(-1)^jn_3)})
	\end{eqnarray}
	So, for any $\rho,$ $I(A:B)$ takes the form:
	\begin{eqnarray}\label{info6}
		I(A:B)&=&\frac{1}{3} \sum_{i=1}^3(H(\frac{1+m_i}{2})-\sum_{j=0}^1\frac{1+(-1)^jn_i}{2}\times\nonumber\\
		&& H(\frac{1-t_i+(-1)^j(n_i-m_i)}{2(1+(-1)^jn_i)}))
	\end{eqnarray}
	\subsection{Details For Single-Qubit Cloning Strategy}
	\subsubsection{Forms Of The States Involved}
	Let $\rho_{AB}$ be an arbitrary two-qubit state distributed from $\mathbf{S}:$
	\begin{eqnarray}\label{arb1}
		\rho_{AB}&=&\sum_{\substack{j,k,\\l,m=0}}^1a_{j,k,l,m}|jk\rangle\langle lm|
	\end{eqnarray}
	$E$ applies cloning operation(Eq.(\ref{cln1})) on $2^{nd}$ qubit of $\rho_{AB}.$
	The state $\rho_{ABE},$ resulting after above operation(Eq.(\ref{cln1})) is applied on $2^{nd}$ qubit of $\rho_{AB},$ takes the form:
	\begin{eqnarray}\label{arb2}
		\rho_{ABE}&=&\sum_{\substack{j,k,\\l,m=0}}^1a_{j,k,l,m}|j\rangle\langle l|\otimes (U_C|k0\rangle\langle m0|U^{\dag}_C)\nonumber\\
		&=&\sum_{\substack{j,k,\\l,m=0}}^1a_{j,k,l,m}|j\rangle\langle l|\otimes(\cos^2\frac{\theta}{2}|k0\rangle\langle m0|+\nonumber\\
		&&k\cdot m\sin^2\frac{\theta}{2}|k\oplus 11\rangle\langle m\oplus 11|+\frac{\sin2\theta}{2}\nonumber\\
		&&(k|k\oplus 11\rangle\langle m0|+m|k0\rangle\langle m\oplus 11|))
	\end{eqnarray}
	Trusted parties $A,B$ execute $\mathcal{P}_{QKD}$ using  $\rho_{AB}^{'}:$
	\begin{eqnarray}\label{arb3}
		\rho_{AB}^{'}&=&\sum_{i=0}^1\prescript{}{E}{\langle i}|\rho_{ABE}|i\rangle_{E}\nonumber\\ &=&\sum_{\substack{\small{j,k,}\\\small{l,m=0}}}^{\small{1}}a_{j,k,l,m}(\cos^2\frac{\theta}{2}|jk\rangle\langle lm|\nonumber\\
		&+&km\sin^2\frac{\theta}{2}|jk\oplus 1\rangle\langle lm\oplus 1|)
	\end{eqnarray}
	$E$ retains the state $\rho_E$ which has the form:
	\begin{eqnarray}\label{arb4}
		\rho_E&=&\sum_{i,j=0}^1\prescript{}{AB}{\langle ij}|\rho_{ABE}|ij\rangle_{AB}\nonumber\\
		&&	\sum_{j,l=0}^1\delta_{j,l}(\sum_{k=0}^12^{k\oplus 1}\cos^2(\frac{\pi k-\alpha}{2})|k\rangle\langle k|\nonumber\\
		&+&\frac{\sin(2\alpha)}{2}|k\rangle\langle k\oplus 1|)
	\end{eqnarray}
%The eigen values($\lambda_1,\lambda_2,$say) of $\rho_E,$  is used to compute its Von-Neumann entropy: $S(\rho_E)$$=$$-$$\sum_{i=0}^1\lambda_i\log_2(\lambda_i).$\\
%The bipartite state that remains shared between $E$ and $A$ is given by:
%	\begin{eqnarray}\label{arb5}
%		\rho_{AE}&=&\prescript{}{B}{\langle}\rho_{ABE}\rangle_{B}\nonumber\\ &=&....
%	\end{eqnarray}
	\subsubsection{Devetak-Winter Key Rate}
	Let us recall the expression of Devetak-Winter key rate(Eq.(\ref{wr7})):
	\begin{eqnarray*}
		r_D&=&\frac{1}{3}\sum_{j=1}^3(I_j(A:B)-\chi_j(A:E))
	\end{eqnarray*}
	Using Eq.(\ref{wr3}), $r_D$ becomes:
	\begin{eqnarray}\label{wa7}
		r_D&=&\frac{1}{3}\sum_{j=1}^3(I_j(A:B)-\chi_j(A:E))\nonumber\\
		&=&\frac{1}{3}\sum_{j=1}^3(I_j(A:B)-S(\rho_E)\nonumber\\
		&+&\sum_{a_j}p(a_j)S(\varrho_E^{(a_j)}))
	\end{eqnarray}
	$\forall j$$=$$1,2,3,$ $I_j(A:B)$ is computed when $A,B$ both measure their subsystems forming $\rho_{AB}^{'}$ in $j^{th}$ MUB. Being a two-qubit state, $\rho_{AB}^{'}$ has a Bloch sphere representation(Eq.(\ref{st41})). Denoting the local Bloch vectors of $\rho_{AB}^{'}$ as $\vec{\mathbf{m}^{'}}$$=$$(m_1^{'},m_2^{'},m_3^{'})$; $\vec{\mathbf{n}^{'}}$$=$$(n_1^{'},n_2^{'},n_3^{'})$; correlation tensor as $T^{'}$$=\textmd{diag}(t_1^{'},t_2^{'},t_3^{'})$ and using Eq.(\ref{info6}), we get:
	\begin{eqnarray}\label{info6i}
		I(A:B)&=&\frac{1}{3} \sum_{i=1}^3(H(\frac{1+m_i^{'}}{2})-\sum_{j=0}^1\frac{1+(-1)^jn_i^{'}}{2}\times\nonumber\\
		&& H(\frac{1-t_i^{'}+(-1)^j(n_i^{'}-m_i^{'})}{2(1+(-1)^jn_i^{'})}))
	\end{eqnarray}
	$S(\rho_E)$ is computed from $\rho_E$(retained by $E$) that does not depend on MUB chosen by $A:$
	\begin{eqnarray}\label{wa9}
		S(\rho_E)&=&-\sum_{i=1}^2\lambda_i\log_2(\lambda_i),
	\end{eqnarray}
	where $\lambda_1,\lambda_2$ are the eigen values of $\rho_E$(Eq.(\ref{arb4})).\\
	%$E$ maintains correlations with the trusted parties through the state $\rho_{ABE}.$ 
	$\forall j$$=$$1,2,3,$ when $A$ measures $\rho_{AE}$ in $j^{th}$ MUB and obtains output $a_j$ with probability $p(a_j),$ $E$ obtains conditional state $\rho_E^{(a_j)}.$ 
	\begin{eqnarray}\label{wa10}
		S(\rho_E^{(a_j)})&=&-\sum_{i=1}^2\lambda_i^{(a_j)}\log_2(\lambda_i^{(a_j)}),
	\end{eqnarray} 
	where $\lambda_1^{(a_j)},\lambda_2^{(a_j)}$ are the eigen values of $\rho_E^{(a_j)}$ with $a_j$$\in$$\{0,1\},$ $\forall j$$=$$1,2,3.$\\
	Using Eqs.(\ref{info6i},\ref{wa9},\ref{wa10}), $r_{D}$Eq.(\ref{wa7}) takes the form:
	\begin{eqnarray}\label{wa7i}
		r_D	&=&\frac{1}{3} \sum_{i=1}^3(H(\frac{1+m_i^{'}}{2})-\sum_{j=0}^1\frac{1+(-1)^jn_i^{'}}{2}\times\nonumber\\
		&& H(\frac{1-t_i^{'}+(-1)^j(n_i^{'}-m_i^{'})}{2(1+(-1)^jn_i^{'})}))+\sum_{i=1}^2\lambda_i\log_2(\lambda_i)\nonumber\\
		&-&\frac{1}{3}\sum_{j=1}^3\sum_{i=1}^2\lambda_i^{(a_j)}\log_2(\lambda_i^{(a_j)})
	\end{eqnarray}
	In particular, Holevo quantity $\chi,$ bounding $E^{'}$ information gain is given by:
	\begin{eqnarray}\label{wa7ii}
		\chi	&=&\frac{1}{3}\sum_{j=1}^3\sum_{i=1}^2\lambda_i^{(a_j)}\log_2(\lambda_i^{(a_j)})-\sum_{i=1}^2\lambda_i\log_2(\lambda_i).
	\end{eqnarray}
\subsubsection{Numerical Illustration}
Let us compute $r_D$ using above forms for a particular numerical example given to illustrate utility of  pure states in the main text. \\
Here $\rho_{AB}$ is a pure state:
\begin{eqnarray}\label{purae}
	\rho_{AB}&=&\frac{1}{2}( |00\rangle\langle 00|+ |11\rangle\langle 11|+ |00\rangle\langle 11|+ |11\rangle\langle 00|).
\end{eqnarray}
In $\mathcal{P}_{QKD}^{C_1}$ characterized by $\theta$$=$$\frac{\pi}{4},$  $\rho_{AB}^{'}$ is given by:
\begin{eqnarray}\label{ill1}
	\rho_{AB}^{'}&=&\frac{1}{2}(|00\rangle\langle 00|+\cos(\frac{\pi}{8})(|00\rangle\langle 11|+|11\rangle\langle 00|)\nonumber\\
	&+&\sin^2(\frac{\pi}{8})|10\rangle\langle 10|+\cos^2(\frac{\pi}{8})|11\rangle\langle 11|)
\end{eqnarray}
From Eq.(\ref{ill1}) we get $\vec{\mathbf{m}^{'}}$$=$$(0,0,0);$ $\vec{\mathbf{n}^{'}}$$=$$(0,0,0.146447)$ and $T$$=$$\textmd{diag}(0.92388,-0.92388,0.853553).$ \\
Using these Bloch parameters in Eq.(\ref{info6i}) we get 
\begin{eqnarray}
	I(A:B)&=& 0.739127.
\end{eqnarray}
Here
\begin{eqnarray}\label{ill2}
	\rho_{AE}&=&\frac{1}{2}(|00\rangle\langle 00|+\sin(\frac{\pi}{8})(|00\rangle\langle 11|+|11\rangle\langle 00|)\nonumber\\
	&+&\cos^2(\frac{\pi}{8})|10\rangle\langle 10|+\sin^2(\frac{\pi}{8})|11\rangle\langle 11|)
\end{eqnarray}

\begin{eqnarray}
	\rho_E&=&\frac{1+\cos^2(\frac{\pi}{8})}{2}|0\rangle\langle 0|+\frac{\sin^2(\frac{\pi}{8})}{2}|1\rangle\langle 1|.
\end{eqnarray}
Eigen values of $\rho_E$ are $\frac{1+\cos^2(\frac{\pi}{8})}{2}, \frac{\sin^2(\frac{\pi}{8})}{2}.$\\
When $A$ measures in $X$ basis:
\begin{eqnarray}
	p(a_1=0)=0.5&;&\rho_E^{(a_1=0)}=\left[ {\begin{array}{cc}
		0.926777	&0.191342\\
		0.191342&0.0732233\\		
	\end{array} } \right]\nonumber\\
	p(a_1=1)=0.5&;&\rho_E^{(a_1=1)}=\left[ {\begin{array}{cc}
			0.926777	&-0.191342\\
			-0.191342&0.0732233\\		
	\end{array} } \right]\nonumber
\end{eqnarray}
Eigen values of each of $\rho_E^{(a_1=0)}$ and  $\rho_E^{(a_1=1)}$ are $0.967707, 0.0322928.$\\
When $A$ measures in $Y$ basis:
\begin{eqnarray}
	p(a_2=0)=0.5&;&\rho_E^{(a_2=0)}=\left[ {\begin{array}{cc}
			0.926777	&0.191342\imath\\
			-0.191342\imath&0.0732233\\		
	\end{array} } \right]\nonumber\\
	p(a_2=1)=0.5&;&\rho_E^{(a_2=1)}=\left[ {\begin{array}{cc}
			0.926777	&-0.191342\imath\\
			0.191342\imath &0.0732233\\		
	\end{array} } \right]\nonumber
\end{eqnarray}
Eigen values of each of $\rho_E^{(a_2=0)}$ and  $\rho_E^{(a_2=1)}$ are again $0.967707, 0.0322928.$\\

When $A$ measures in $Z$ basis:
\begin{eqnarray}
	p(a_3=0)=0.5&;&\rho_E^{(a_3=0)}=\left[ {\begin{array}{cc}
		1	&0\\
			0&0\\		
	\end{array} } \right]\nonumber\\
	p(a_3=1)=0.5&;&\rho_E^{(a_3=1)}=\left[ {\begin{array}{cc}
		0.853553	&0\\
			0 &0.146447\\		
	\end{array} } \right]\nonumber
\end{eqnarray}
Eigen values of $\rho_E^{(a_3=0)}$ are $1,0$ and that of  $\rho_E^{(a_3=1)}$ are $0.853553, 0.146447.$\\
Using all the above details in Eq.(\ref{wa7i}) and Eq.(\ref{wa7ii}) we get 
\begin{eqnarray}
	r_D&=&0.59861\nonumber\\
	\chi&=&0.140517.
\end{eqnarray}
	\subsection{Proof Of Theorem.\ref{clo1}}
Here $\rho_{AB}$ is a two-qubit pure state(Eq.(\ref{pure})). $\rho_{AB}^{'}$(Eq.(\ref{arb3})) is given by:
\begin{eqnarray}\label{theo11}
	\rho_{AB}^{'}&=&\cos^2\alpha|00\rangle\langle 00|+\frac{1}{2}\cos(\frac{\theta}{2})\sin(2\alpha)(|00\rangle\langle 11|\nonumber\\
	&+&|11\rangle\langle 00|)+\sin^2\alpha(\sin(\frac{\theta}{2})|01\rangle\langle 01|\nonumber\\
	&+&\cos(\frac{\theta}{2})|10\rangle\langle 10|)
\end{eqnarray}
Bloch parameters of Eq.(\ref{theo11}) are:
\begin{eqnarray}\label{theo12}
	\vec{\mathbf{m}^{'}}&=&(0,0\cos(2\alpha));\nonumber\\
	\vec{\mathbf{n}^{'}}&=&(0,0\cos^2(\alpha)-\sin^2(\alpha)\cos\theta);\nonumber\\
	T^{'}&=&\textmd{\small{diag}}(\cos(\frac{\theta}{2})\sin(2\alpha),-\cos(\frac{\theta}{2})\sin(2\alpha),\nonumber\\
	&&\cos^2\alpha+\sin^2\alpha\cos\theta)
\end{eqnarray}
Using Eq.(\ref{theo12}), simplified form of $I(A:B)$ is given by:
\begin{eqnarray}\label{theor12}
	I(A:B)&=&\frac{1}{3}(2+H(\cos^2\alpha)\nonumber\\
	&-&\sum_{j=0}^1H(\frac{1+(-1)^j\sin(2\alpha)\cos(\frac{\theta}{2})}{2})\nonumber\\
	&-&(\cos^2\alpha\cos^2(\frac{\theta}{2})+\sin^2(\frac{\theta}{2}))\times\nonumber\\
	&&H(\frac{\sin^2\alpha\sin^2(\frac{\theta}{2})}{\cos^2\alpha\cos^2(\frac{\theta}{2})+\sin^2(\frac{\theta}{2})})) 
\end{eqnarray}
$\rho_E^{'}$s eigen values are $\cos^2\alpha+\cos^2(\frac{\theta}{2})\sin^2\alpha$ and $\sin^2(\frac{\theta}{2})\sin^2\alpha.$ So, 
\begin{eqnarray}\label{theoirs12}
	S(\rho_E)&=&H(\sin^2(\frac{\theta}{2})\sin^2\alpha)
\end{eqnarray}
\textit{When $A$ measures in $X$ basis:}
\begin{eqnarray}\label{theo13}
	P(a_1=j)&=&\frac{1}{2},\,\, j=0,1\nonumber\\
	\textmd{\small{Eigen values of }}\rho_E^{(a_1=0)}&=&\frac{8\pm\sqrt{2A}}{16}\nonumber\\
		\textmd{\small{Eigen values of }}\rho_E^{(a_1=1)}&=&\frac{8\pm\sqrt{2A}}{16}\nonumber\\
	\textmd{\small{with} }A&=&26+8\cos(2\alpha)-2\cos(4\alpha)\nonumber\\\
+\sum_{j=0}^1(-4\cos(2\alpha+(-1)^j2\theta)	&+&\cos(4\alpha+(-1)^j2\theta))\nonumber\\
&&+6\cos(2\theta).
\end{eqnarray}
\textit{When $A$ measures in $Y$ basis:}
\begin{eqnarray}\label{theo14}
	P(a_2=j)&=&\frac{1}{2},\,\, j=0,1\nonumber\\
	\textmd{\small{Eigen values of }}\rho_E^{(a_2=0)}&=&\frac{8\pm\sqrt{2A}}{16}\nonumber\\
		\textmd{\small{Eigen values of }}\rho_E^{(a_2=1)}&=&\frac{8\pm\sqrt{2A}}{16}\nonumber\\
	\textmd{\small{where} }A&&	\textmd{\small{is given by Eq.(\ref{theo13}).}}\nonumber\\	
\end{eqnarray}
\textit{When $A$ measures in $Z$ basis:}
\begin{eqnarray}\label{theo15}
	P(a_3=0)=\cos^2\alpha&;&P(a_3=1)=\sin^2\alpha\nonumber\\
	\textmd{\small{Eigen values of }}\rho_E^{(a_3=0)}&=&1,0\nonumber\\
	\textmd{\small{Eigen values of }}\rho_E^{(a_3=1)}&=&\cos^2(\frac{\theta}{2}),\,\sin^2(\frac{\theta}{2})
\end{eqnarray}
Using Eqs.(\ref{theor12}-\ref{theo15}), from Eq.(\ref{wa7i}), we get:
\begin{eqnarray}\label{theo16}
	r_D&=&f_1(\theta,\alpha)\,\,\textmd{\small{where}}\nonumber\\
	f_1(\theta,\alpha)&=&-H(\sin^2(\frac{\theta}{2})\sin^2\alpha)\nonumber\\
	&+&\frac{1}{3}(2+H(\cos^2\alpha)\nonumber\\
	&-&\sum_{j=0}^1H(\frac{1+(-1)^j\sin(2\alpha)\cos(\frac{\theta}{2})}{2})\nonumber\\
	&-&(\cos^2\alpha\cos^2(\frac{\theta}{2})+\sin^2(\frac{\theta}{2}))\times\nonumber\\
	&&H(\frac{\sin^2\alpha\sin^2(\frac{\theta}{2})}{\cos^2\alpha\cos^2(\frac{\theta}{2})+\sin^2(\frac{\theta}{2})}))\nonumber\\
	%&-&\nonumber\\
	&+&\frac{2}{3}H(\frac{8+\sqrt{2A}}{16})+\frac{\sin^2\alpha}{3}H(\cos^2(\frac{\theta}{2})),\nonumber\\
	\textmd{\small{with}}&&\textmd{\small{ A provided by Eq.\ref{theo14}}}.
\end{eqnarray}
Now for any value of $\theta,\alpha$$\in$$[0,\frac{\pi}{4}],$ $f_1(\theta,\alpha)$ is positive(see Fig.\ref{fig2}). So $r_D$$>$$0.$ This in turn proves the theorem.$\blacksquare$
	\begin{center}
	\begin{figure}
		\includegraphics[width=3in]{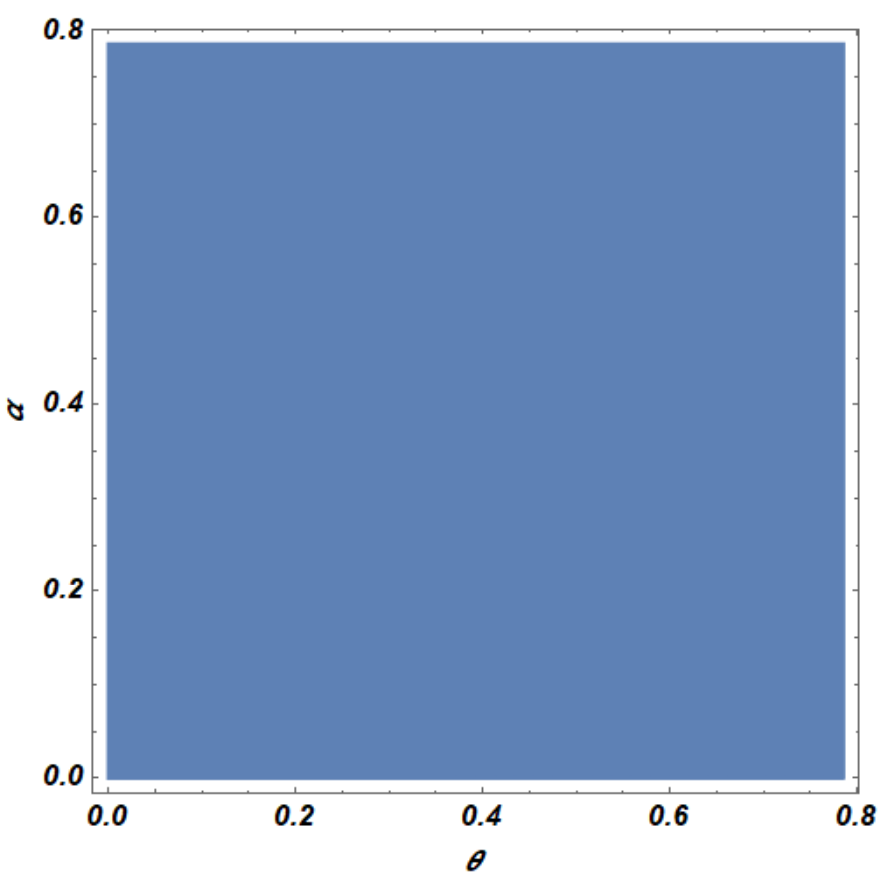} 
		\caption{\emph{At any point $(\theta,\alpha)$ in the shaded region,  $f_1(\theta,\alpha),$ given in Eq.(\ref{theo16}) is positive. Clearly the shaded region covers the entire parameter space formed by the parameter($\alpha$) of the pure state(Eq.(\ref{pure})) and the parameter($\theta$) characterizing the cloning operation(Eq.(\ref{cln1})).}}
		\label{fig2}
	\end{figure}
\end{center}
\subsection{Proof of Result Involving Werner Class of States}
This result can be proved following similar strategy as that followed for Theorem.\ref{clo1}.
Here $\rho_{AB}$ is the Werner family of states(Eq.(\ref{werner})). 
After cloning state shared between $A,B$ is given by:
\begin{eqnarray}\label{theoo11}
	\rho_{AB}^{'}&=&\frac{1}{4}(1-v+(1+v)\sin^2(\frac{\theta}{2})|00\rangle\langle 00|\nonumber\\
	&+&(1+v)\cos^2(\frac{\theta}{2})|01\rangle\langle 01|+(1-v)\sin^2(\frac{\theta}{2})|10\rangle\langle 10|\nonumber\\
	&+&(1-v)\cos^2(\frac{\theta}{2})|11\rangle\langle 11|)\nonumber\\
	&-&\frac{v\cos(\frac{\theta}{2})}{2}(|01\rangle\langle 10|+|10\rangle\langle 01|).
\end{eqnarray}
Bloch parameters of Eq.(\ref{theoo11}) are:
\begin{eqnarray}\label{theoo12}
	\vec{\mathbf{m}^{'}}&=&(0,0,0);\nonumber\\
	\vec{\mathbf{n}^{'}}&=&(0,0,\sin^2(\frac{\theta}{2}));\nonumber\\
	T^{'}&=&\textmd{\small{diag}}(-v\cos(\frac{\theta}{2}),-v\cos(\frac{\theta}{2}),-v\cos^2(\frac{\theta}{2})).\nonumber\\
\end{eqnarray}
Using Eq.(\ref{theoo12}), simplified form of $I(A:B)$ is given by:
\begin{eqnarray}\label{theoor12}
	I(A:B)&=&\frac{1}{3}(3-2H(\frac{1+v\cos(\frac{\theta}{2})}{2})-\frac{\cos^2(\frac{\theta}{2})}{2}H(\frac{1+v}{2})\nonumber\\
	&-&\frac{1+\sin^2(\frac{\theta}{2})}{2}H(\frac{1}{2}+\frac{v\cos^2(\frac{\theta}{2})}{2(1+\sin^2(\frac{\theta}{2}))}).
\end{eqnarray}
$\rho_E^{'}$s eigen values are $\frac{1+\cos^2(\frac{\theta}{2})}{2}$ and $\frac{\sin^2(\frac{\theta}{2})}{2}.$ So, 
\begin{eqnarray}\label{theors12}
	S(\rho_E)&=&H(\frac{\sin^2(\frac{\theta}{2})}{2})
\end{eqnarray}
\textit{When $A$ measures in $X$ basis:}
\begin{eqnarray}\label{theor13}
	P(a_1=j)&=&\frac{1}{2},\,\, j=0,1\nonumber\\
	\textmd{\small{Eigen values of }}\rho_E^{(a_1=0)}&=&\frac{1\pm\sqrt{B}}{2}\nonumber\\
	\textmd{\small{Eigen values of }}\rho_E^{(a_1=1)}&=&\frac{1\pm\sqrt{B}}{2}\nonumber\\
	\textmd{\small{with} }B&=&v^2\sin^2(\frac{\theta}{2})+\cos^4(\frac{\theta}{2})
\end{eqnarray}
\textit{When $A$ measures in $Y$ basis:}
\begin{eqnarray}\label{theor14}
	P(a_2=j)&=&\frac{1}{2},\,\, j=0,1\nonumber\\
	\textmd{\small{Eigen values of }}\rho_E^{(a_2=0)}&=&\frac{1\pm\sqrt{B}}{2}\nonumber\\
	\textmd{\small{Eigen values of }}\rho_E^{(a_2=1)}&=&\frac{1\pm\sqrt{B}}{2}\nonumber\\
	\textmd{\small{with} }B&=&v^2\sin^2(\frac{\theta}{2})+\cos^4(\frac{\theta}{2})
\end{eqnarray}

\textit{When $A$ measures in $Z$ basis:}
\begin{eqnarray}\label{theor15}
	P(a_3=0)&=&P(a_3=1)=\frac{1}{2}\nonumber\\
	\textmd{\small{Eigen values of }}\rho_E^{(a_3=0)}&=&\frac{3-v+\cos\theta+v\cos\theta}{4},\nonumber\\
	&&\frac{(1+v)\sin^2(\frac{\theta}{2})}{2}\nonumber\\
	\textmd{\small{Eigen values of }}\rho_E^{(a_3=1)}&=&\frac{3+v+\cos\theta-v\cos\theta}{4},\nonumber\\
	&&\frac{(1-v)\sin^2(\frac{\theta}{2})}{2}
\end{eqnarray}
Using Eqs.(\ref{theor12}-\ref{theor15}), from Eq.(\ref{wa7i}), we get:
\begin{eqnarray}\label{theor16}
	r_D&=&f_2(\theta,v)\,\,\textmd{\small{where}}\nonumber\\
	f_2(\theta,v)&=&-H(\frac{\sin^2(\frac{\theta}{2})}{2})+\frac{1}{3}(3-2\times\nonumber\\
	&&H(\frac{1+v\cos(\frac{\theta}{2})}{2})-\frac{\cos^2(\frac{\theta}{2})}{2}H(\frac{1+v}{2})\nonumber\\
	&-&\frac{1+\sin^2(\frac{\theta}{2})}{2}H(\frac{1}{2}+\frac{v\cos^2(\frac{\theta}{2})}{2(1+\sin^2(\frac{\theta}{2}))})\nonumber\\
	&+&2H(\frac{1+\sqrt{B}}{2})+\frac{1}{2}\sum_{j=0}^1 H(\frac{(1+(-1)^j v)\sin^2(\frac{\theta}{2})}{2}))\nonumber\\
		\textmd{\small{with}}&&\textmd{\small{B provided by Eq.\ref{theor14}}}.
\end{eqnarray}
Now for any value of $\theta$$\in$$[0,\frac{\pi}{4}]$ and $v$$\in$$[0,1],$ $f_2(\theta,v)$ is positive(see Fig.\ref{fig3}). So $r_D$$>$$0.$ This in turn proves the result.
\begin{center}
	\begin{figure}
		\includegraphics[width=2.8in]{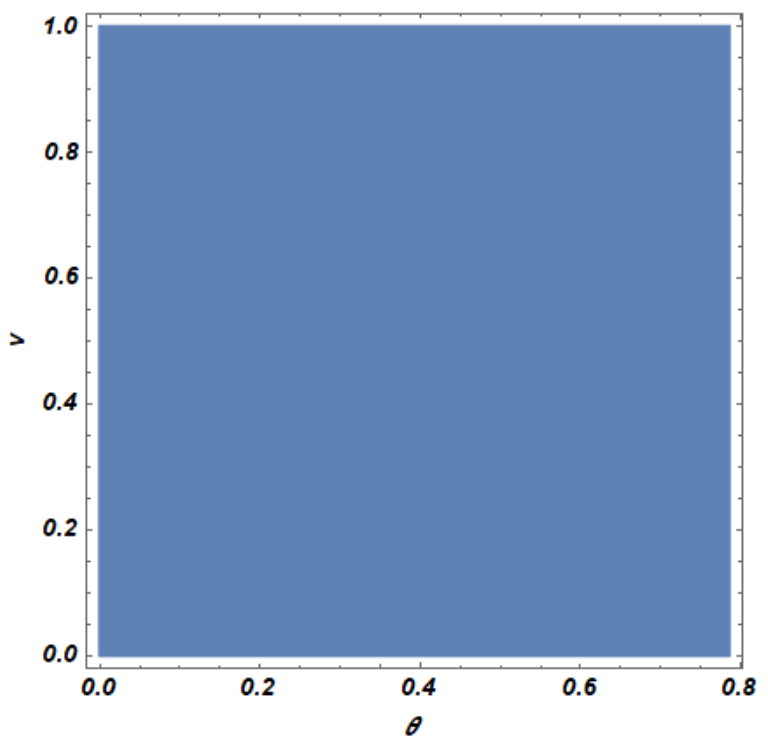} 
		\caption{\emph{At any point $(\theta,v)$ in the shaded region,  $f_2(\theta,v),$ given in Eq.(\ref{theor16}) is positive. Clearly the shaded region covers the entire parameter space formed by the parameter($v$) of Werner state(Eq.(\ref{werner})) and the parameter($\theta$) characterizing the cloning operation(Eq.(\ref{cln1})).}}
		\label{fig3}
	\end{figure}
\end{center}

\subsection{Details For Both Qubit Cloning Strategy}
	\subsubsection{Forms Of The States Involved}
	Here we discuss the detailed forms of the states used in $\mathcal{P}_{QKD}^{C_2}.$\\
	Let $U_C^{(\theta_1)},U_C^{(\theta_2)}$ be the cloners(Eq.(\ref{cln1})) applied by $E$ on $1^{st}$ and $2^{nd}$ qubits transmitted from $\mathbf{S}:$ 
	\begin{eqnarray}\label{clnb1}
		U_C^{(\theta_j)}(|i0\rangle)&=&\cos\frac{\theta_j}{2}	|i0\rangle+i\sin\frac{\theta_j}{2} |i\oplus 1 0\rangle,\\
		&&\quad \quad \theta_j\in[0,\frac{\pi}{4}];\, \,j=1,2.\nonumber
	\end{eqnarray}	
	After $E$ applies cloning operations(Eq.(\ref{clnb1})) on both qubits of an arbitrary two-qubit state $\rho_{AB}$(Eq.(\ref{arb1})) resulting $4$-qubit state $\rho_{AE_1BE_2}:$
	\begin{eqnarray}
		\rho_{AE_1BE_2}=\otimes_{j=1}^2U_C^{(\theta_j)}( \rho_{AB})\cdot (\otimes_{j=1}^2U_C^{(\theta_j)})^{\dagger}&&\nonumber\\
		=\otimes_{j=1}^2U_C^{(\theta_j)}(\sum_{\substack{i,j,\\k,l=0}}^1 a_{i,j,k,l}|ij\rangle\langle kl|)\cdot (\otimes_{j=1}^2U_C^{(\theta_j)})^{\dagger}&&\nonumber\\
		=\sum_{\substack{i,j,\\k,l=0}}^1a_{i,j,k,l}(
		U_C^{(\theta_1)}(|i\rangle\langle k|)(U_C^{(\theta_1)})^{\dagger}\otimes 
		U_C^{(\theta_2)}(|j \rangle\langle l|)(U_C^{(\theta_2)})^{\dagger} )&&\nonumber\\
		=\sum_{\substack{i,j,\\k,l=0}}^1a_{i,j,k,l}(\cos^{i+k}\frac{\theta_1}{2}|i0\rangle\langle k0|+i \cos^{k}\frac{\theta_1}{2}\sin\frac{\theta_1}{2}|i\oplus 1 i\rangle\langle k0|&&\nonumber\\
		+k \cos^i\frac{\theta_1}{2}\sin\frac{\theta_1}{2}|i0\rangle\langle k\oplus 1 k|+ik \sin^2\frac{\theta_1}{2}|i\oplus 1 i\rangle\langle k\oplus 1 k|)\otimes&&\nonumber\\
		(\sum_{\substack{j,\\l=0}}^1[\cos^{j+l}\frac{\theta_2}{2}|j0\rangle\langle l0|+j \cos^{l}\frac{\theta_2}{2}\sin\frac{\theta_2}{2}|j\oplus 1 j\rangle\langle l0|&&\nonumber\\
		+l \cos^j\frac{\theta_2}{2}\sin\frac{\theta_2}{2}|j0\rangle\langle l\oplus 1 l|+jl \sin^2\frac{\theta_2}{2}|j\oplus 1 j\rangle\langle l\oplus 1 l|)&&\nonumber\\
	\end{eqnarray}
	Here $E$ retains a two-qubit state $\rho_{E_1E_2}$$=$$\sum_{m,n=0}^1\prescript{}{AB}{\langle} mn|\cdot \rho_{AE_1BE_2}\cdot|mn\rangle_{AB}$ whereas the trusted parties share another two-qubit state $\rho_{AB}^{''}$$=$$\sum_{m,n=0}^1\prescript{}{E_1E_2}{\langle} mn|\cdot \rho_{AE_1BE_2}\cdot|mn\rangle_{E_1E_2}.$
\subsubsection{Devetak-Winter Key Rate}
Proceeding as in the case of single-qubit cloning strategy, we get the following expression of Devetak-Winter Key Rate $r_D:$ 
	\begin{eqnarray}\label{wd7i}
	r_D	&=&\frac{1}{3} \sum_{i=1}^3(H(\frac{1+m_i^{''}}{2})-\sum_{j=0}^1\frac{1+(-1)^jn_i^{''}}{2}\times\nonumber\\
	&& H(\frac{1-t_i^{''}+(-1)^j(n_i^{''}-m_i^{''})}{2(1+(-1)^jn_i^{''})}))+\sum_{i=1}^4\nu_i\log_2(\nu_i)\nonumber\\
	&-&\frac{1}{3}\sum_{j=1}^3\sum_{i=1}^2\nu_i^{(a_j)}\log_2(\nu_i^{(a_j)})
\end{eqnarray}
and Holevo quantity $\chi$:
\begin{eqnarray}\label{wd7ii}
	\chi	&=&\frac{1}{3}\sum_{j=1}^3\sum_{i=1}^4\nu_i^{(a_j)}\log_2(\nu_i^{(a_j)})-\sum_{i=1}^4\nu_i\log_2(\nu_i).
\end{eqnarray}
In above forms:
\begin{itemize}
	\item $\vec{\mathbf{m}^{''}}$$=$$(m_1^{''},m_2^{''},m_3^{''})$; $\vec{\mathbf{n}^{''}}$$=$$(n_1^{''},n_2^{''},n_3^{''})$ are local Bloch vectors of $\rho_{AB}^{''}$(Eq.(\ref{rhoab}))
	\item correlation tensor as $T^{''}$$=\textmd{diag}(t_1^{''},t_2^{''},t_3^{''})$ is correlation tensor of $\rho_{AB}^{''}$(Eq.(\ref{rhoab}))
	\item $\nu_1,\cdots,\nu_4$ are eigen values of two-qubit state $\rho_{E_1E_2}$(Eq.(\ref{rhoe1e2})).
	\item $\nu_1^{(a_j)},\cdots,\nu_4^{(a_j)}$ are the eigen values of $E^{'}$s two-qubit conditional state $\rho_{E_1E_2}^{(a_j)}$ conditioned on $A$ obtaining output $a_j$(with probability $p(a_j)$) while measuring her local subsystem in $j^{th}$ basis. 
\end{itemize}

	\begin{comment}
		Ideally pure entangled state 
		fails any one of the validation checks. However, even if $\mathcal{P}_{QKD}$ passes both the checks, it can still 
		The state($\rho_E$) retained by $E$ is given by:
		\begin{eqnarray}
			\rho_E&=&
		\end{eqnarray}
		The two stages of one-way key distillation can therefore be viewed as
		\begin{itemize}
			\item \textbf{Information reconciliation(error correction):}
			Alice communicates suitable classical information to Bob so that
			Bob can correct discrepancies between their raw keys. This incurs
			an information cost, characterized by the conditional entropy
			\(H(A|B)\).

			\item \textbf{Privacy amplification:}
			Alice and Bob apply a suitable public hash function to their
			reconciled keys, shortening them so that Eve's residual information
			is sufficiently suppressed.

		\end{itemize}
	\end{comment}
\end{document}